\documentclass[conference]{IEEEtran}

\usepackage{times}

\usepackage{amsmath}
\usepackage{amsfonts}
\usepackage{latexsym}
\usepackage{amssymb}

\usepackage{upref}
\usepackage{theorem}
\usepackage{graphicx}
\usepackage{url}
\usepackage{subfigure}

\usepackage[left=2.5cm,top=2.5cm,right=2.5cm,bottom=1.5cm,]{geometry}

\theoremstyle{plain} \theorembodyfont{\normalfont\slshape}

\newtheorem{thm}{Theorem$\!$}
\newenvironment{theorem}
{\begin{thm}\hspace*{-1ex}{\bf.}}{\end{thm}}

\newtheorem{lem}[thm]{Lemma$\!$}

\newtheorem{prop}[thm]{Proposition$\!$}

\newtheorem{cor}[thm]{Corollary$\!$}

\newtheorem{defn}[thm]{Definition$\!$}

\newtheorem{xmpl}[thm]{Example$\!$}
\newenvironment{example}{\begin{xmpl}\hspace*{-1ex}{\bf.}}{\end{xmpl}}

\newtheorem{cnstr}[thm]{Construction$\!$}

\newtheorem{algr}[thm]{Algorithm$\!$}

\begin{document}

\title {Rebuilding for Array Codes in Distributed Storage Systems}
\author{
\authorblockN{Zhiying Wang}
\authorblockA{Electrical Engineering Department\\
California Institute of Technology \\
Pasadena, CA 91125\\
Email: zhiying@caltech.edu}
\and
\authorblockN{Alexandros G. Dimakis}
\authorblockA{Electrical Engineering Department\\
University of Southern California\\
 Los Angeles, CA 90089-2560\\
 Email: dimakis@usc.edu\\}
\and
\authorblockN{Jehoshua Bruck}
\authorblockA{Electrical Engineering Department\\
California Institute of Technology\\
Pasadena, CA 91125\\
Email: bruck@caltech.edu} }
\date{}
\maketitle

\begin{abstract}
In distributed storage systems that use coding, the issue of minimizing the communication required to rebuild 
a storage node after a failure arises. 
We consider the problem of repairing an erased node in a distributed storage system that uses an EVENODD code. 
EVENODD codes are maximum distance separable (MDS) array codes that are used to protect against erasures, and only require XOR operations for encoding and decoding. 
We show that when there are two redundancy nodes, to rebuild one erased systematic node, only $3/4$ of the information needs to be transmitted. Interestingly, in many cases, the required disk I/O is also minimized. 
\end{abstract}

\section{Introduction}

Coding techniques for storage systems have been used widely to protect data
against errors or erasure for CDs, DVDs, Blu-ray Discs, and SSDs. Assume
the data in a storage system is divided into packets of equal sizes. An
$(n,k)$ block code takes $k$ information packets and encodes them into a
total of $n$ packets of the same size. Among coding schemes, maximum distance
separable (MDS) codes offer maximal reliability for a given redundancy: any $k$ packets 
are sufficient to retrieve all the information. Reed-Solomon codes~\cite{RS} are the most well known MDS codes that  
are used widely in storage and communication applications. 
Another class of MDS codes are MDS array codes, for example EVENODD~\cite{evenodd} and its extension \cite{evenoddex}, B-code \cite{Bcode}, X-code \cite{xcode}, RDP \cite{RDP}, and STAR code \cite{star}. In an array code, each of the packets consists of a column of elements (one or more binary bits), and the parities are computed by XORing some information bits. These codes have the advantage of low computational complexity over RS codes because the encoding and decoding only involve XOR operations.

Distributed storage systems involving storage nodes connected over networks have recently 
attracted a lot of attention. MDS codes can be used for erasure protection in distributed storage systems
where encoded information is stored in a distributed manner. If no more than $n-k$ storage nodes are lost, then all the information can still be recovered from the surviving packets. Suppose one packet is erased, and instead of retrieving the entire $k$ packets of information, if we are only interested in repairing the lost packet, then what is smallest amount of transmission needed (called the \emph{repair bandwidth})? If we transmit $k$ packets from the other nodes to the erased one, then by the MDS property, we can certainly repair this node. But can we transmit less than $k$ packets? More generally, if no more than $n-k$ nodes are erased, what is the  repair bandwidth? This \emph{repair problem} was first raised in \cite{bound}, and was further studied in several works (e.g. \cite{repair1}-\nocite{repair2}\nocite{repair3}\nocite{repair4}\nocite{repair5}\cite{repair6}). A recent survey of this problem can be found in \cite{survey}. In \cite{bound}, a cut-set lower bound for repair bandwidth is derived and in \cite{repair3}\cite{repair4}\cite{repair5}, this lower bound is matched for exact repair by code constructions for $k=2,3$, $n-1$ and $2k \le n$. All of these constructions however require large finite fields.
Very recently it was established that the cut-set bound of~\cite{bound} is achievable for all values of $k$ and $n$, \cite{repair5}\cite{repair6}. However, the proof is theoretical and is based on very large finite fields. Hence, it does not provide the basis for constructing practical codes with small finite fields and high rate.

In this paper we take a different route: rather than trying to construct MDS codes that are easily repairable,
we try to find ways to repair existing codes and specifically focus on the families of MDS array codes.  
A related and independent work can be found in \cite{optRDP}, where single-disk recovery for RDP code was studied, and the recovery method and repair bandwidth is indeed similar to our result. Besides, \cite{optRDP} discussed balancing disk I/O reads in the recovery. Our work discusses the recovery of single or double disk recovery for EVENODD, X-code, STAR, and RDP code.  

If the whole data object stored has size $M$ bits, repairing a single erasure naively would require communicating (and reading) $M$ bits from surviving storage nodes. Here we show that a single failed 
systematic node can be rebuilt after communicating only $\frac{3}{4}M+ O(M^{1/2})$ bits. 
Note that the cut-set lower bound~\cite{bound} scales like $\frac{1}{2}M+O(M^{1/2})$, so 
it remains open if the repair communication for EVENODD codes can be further reduced. 
Interestingly our repair scheme also requires significantly less disk I/O reads compared to naively reading 
the whole data object. 

The rest of this paper is organized as follows. In Section \ref{def}, we are going to define EVENODD code and the repair problem. Then the repair of one lost node is presented in Section \ref{section evenodd} for EVENODD ($k=n-2$) and in Section \ref{section ex} for the extended EVENODD ($k<n-2$). In Section \ref{section 2erase}, we consider the case with two erased nodes and $k=n-3$. At last, conclusion is made in Section \ref{conclusion}.

\section{Definitions} \label{def}
An $R\times n$ array code contains $R$ rows and $n$ columns (or packets). Each element in the array can be a single bit or a block of bits. We are going to call an element a \emph{block}. In an $(n,k)$ array code, $k$ information columns, or systematic columns, are encoded into $n$ columns. The \emph{total amount of information} is $M=Rk$ blocks.

An EVENODD code \cite{evenodd} is a binary MDS array code that can correct up to 2 column erasures. For a prime number $p \ge 3$, the code contains $R=p-1$ rows and $n=p+2$ columns, where the first $k=p$ columns are information and the last two are parity. And the information is $M=(p-1)p$ blocks.

We will write an EVENODD code as:
$$\begin{array}{llllll}
a_{1,1} & a_{1,2} & \dots & a_{1,p} & b_{1,0} & b_{1,1}\\
a_{2,1} & a_{2,2} & \dots & a_{2,p} & b_{2,0} & b_{2,1}\\
\vdots & \vdots & 				& \vdots & \vdots &\vdots \\
a_{p-1, 1} & a_{p-1, 2} & \dots & a_{p-1, p} & b_{{p-1},0} & b_{{p-1},1}\\
\end{array}$$
And we define an imaginary row $a_{p,j}=0$, for all $j=1,2,\dots, p$, where $0$ is a block of zeros. The \emph{slope 0} or \emph{horizontal parity} is defined as 
\begin{equation} \label{eqn1}
b_{i,0} = \sum_{j=1}^{p} a_{i,j}
\end{equation}
for $i = 1,\dots,p-1$. The addition here is bit-by-bit XOR for two blocks. 
A parity block of \emph{slope} $v$, $-p < v < p$ and $v \neq 0$ is defined as
\begin{equation}\label{eq10}
b_{i,v}=\sum_{j=1}^{p} a_{j,<i+v(1-j)>}+S_v = \sum_{j=1}^{p} a_{<i+v(1-j)>,j} + S_v
\end{equation}
where 
$S_v = a_{p,1}+a_{p-v,2}+\dots+a_{<p+v>,p} = \sum_{j=1}^{p}a_{<v(1-j)>,j}$ 
and $<x> = (x-1)$ mod $p$ $+1$. Sometimes we omit the ``$<>$'' notation. When $v=1$, we call it the \emph{slope 1}, or \emph{diagonal parity}. In EVENODD, parity columns are of slopes 0 and 1.

A similar code is RDP \cite{RDP}, where $R=p-1$, $n=p+1$, and $k=p-1$, for a prime number $p$. The diagonal parity sums up both the corresponding information blocks and one horizontal parity block. Another related code is X-code \cite{xcode}, where the parity blocks are of slope -1 and 1, and are placed as two additional rows, instead of two parity columns.

The code in \cite{evenoddex} extended EVENODD to more than 2 columns of parity. This code has $n=p+r$, $k=p$, and $R=p-1$. The information columns are the same as EVENODD, but $r$ parity columns of slopes $0,1,\dots,r-1$ are used. 
It is shown in \cite{evenoddex} that such a code is MDS when $r \le 3$ and conditions for a code to be MDS are derived for $r \le 8$. 

STAR code \cite{star} is an MDS array code with $k=p,R=p-1,n=p+3$, and the parity columns are of slope 0, 1, and -1. 

A \emph{parity group} $B_{i,v}$ of slope $v$ contains a parity block $b_{i,v}$ and the information blocks in the sum in equations (\ref{eqn1}) (\ref{eq10}), $i = 1,2,\dots,p-1 $. $S_v$ is considered as a single information block. If $v=0$, it is a \emph{horizontal parity group}, and if $v=1$, we call it a \emph{diagonal parity group}. 

By (\ref{eqn1}), each horizontal parity group $B_{i,0}$ contains $a_{i,<k+1-i>}$ $\in$ $B_{k,1}$, for all $k=1,2,\dots,p-1$. So we say $B_{i,0}$ \emph{crosses with} $B_{k,1}$, for all $k=1,2,\dots,p-1$.  Conversely, each diagonal parity group $B_{i,1}$ contains $a_{k,<i+1-k>}$ $\in$ $B_{k,0}$, for all $k=1,2,\dots,p-1$. Therefore, $B_{i,1}$ crosses with $B_{k,0}$ for all $k=1,2,\dots,p-1$. The shared block of two parity groups is called the \emph{crossing}. Generally, two parity groups $B_{i,v}$ and $B_{k,u}$ cross, for $v \neq u$, $1 \le i,k \le p-1$. If they cross at $a_{p,<i+v>}=0$, we call it a \emph{zero crossing}. A zero crossing does not really exist since the $p$-th row is imaginary. A zero crossing occurs if and only if 
\begin{equation} \label{i+v}
u,v \neq 0 \textrm{ and } <i+v> = <k+u>
\end{equation}
Moreover, each information block belongs to only one parity group of slope $v$.

Suppose the $n$ packets are stored in $n$ different nodes in a connected network. Each storage node contains exactly one packet (or one column). 
Assume $n-d$ nodes are erased, $d \ge k$. Suppose we recover the nodes successively. For any specified erased node, how many blocks from the other storage nodes are needed to recover it? We can either send data in a single block, or a linear combination of several blocks in one node, both of which are counted as one block of transmission. The total number of blocks transmitted to recover the specified node is called the \emph{repair bandwidth} $\gamma$. The \emph{repair problem} for distributed storage system asks what the smallest $\gamma$ is, for fixed $M,d,k$. In \cite{bound}, a cut-set lower bound is derived (and is achieved only when each node transmits the same number of blocks):
\begin{equation} \label{eq3}
\gamma^* = \frac{Md}{k(d-k+1)}
\end{equation}

In this paper, we use MDS array codes as distributed storage codes. We will give repair methods and compute the corresponding bandwidth $\gamma$.

\begin{example}\label{example1}
Consider the EVENODD code with $p=3$. Set $a_{1,3}=a_{2,3}=0$ for all codewords, then the code will contain only 2 columns of information. The resulting code is a $(4,2)$ MDS code and this is called \emph{shortened} EVENODD (see Figure \ref{fig1}). It can be verified that if any node is erased, then sending 1 block from each of the other nodes is sufficient to recover it. And this actually matches the bound (\ref{eq3}). Figure \ref{fig1} shows how to recover the first or the fourth column. Notice that a sum block is sent in some cases. For instance, to recover the first column, the sum $b_{1,1}+b_{2,1}$ is sent from the fourth column.
\end{example}

\begin{figure}
	\centering
		\includegraphics[width=0.4\textwidth]{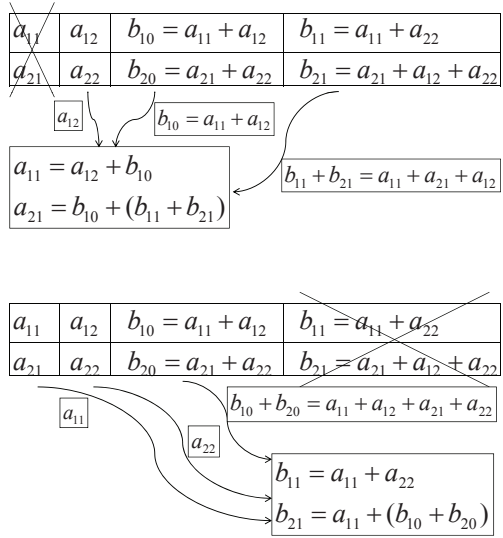}
	\caption{Repair of a $(4,2)$ EVENODD code if the first column (top graph) or the fourth column (bottom graph) is erased. In both cases, three blocks are transmitted.}
	\label{fig1}
	\vspace{-0.5cm}
\end{figure}

In this paper, shortening of a code is not considered and we will focus on the recovery of systematic nodes, given that 1 or 2 systematic nodes are erased. And we send no linear combinations of data except the sum $\sum_{i=1}^{p-1} b_{i,v}$ from the parity node of slope $v$, for all $v$ defined in an array code. In addition, we assume that each node can transmit a different number of blocks.

\section{Repair for Codes with 2 Parity Nodes} \label{section evenodd}
First, let us consider the repair problem of losing one systematic node, $n-d=1$, and $n-k=2$. We will use EVENODD to explain the repair method, and the recovery will be very similar if RDP or X-code is considered.

By the symmetry of the code, we assume that the first column is missing. Each block in the first column must be recovered through either the horizontal or the diagonal parity group including this block. Suppose we use $x$ horizontal parity groups and $p-1-x$ diagonal parity groups to recover the column, $0 \le x \le p-1$. These parity groups include all blocks of the first column exactly once. 

Notice that $S_1 = \sum_{i=1}^{p-1} b_{i,0} + \sum_{i=1}^{p-1} b_{i,1}$, so we can send $\sum_{i=1}^{p-1} b_{i,0}$ from the $(p+1)$-th node, and $\sum_{i=1}^{p-1} b_{i,1}$ from the $(p+2)$-th node, and recover $S_1$ with 2 blocks of transmission. For the discussion below, assume $S_1$ is known.

For each horizontal parity group $B_{i,0}$, we send $b_{i,0}$ and $a_{i,j}$, $j=2,3,\dots,p$. So we need $p$ blocks. 
For each diagonal parity group $B_{i,1}$, as $S_1$ is known, we send $b_{i,1}$ and $a_{j,<i+1-j>}$, $j=1,2,\dots,i-1,i+1,\dots,p-1$, which is $p-1$ blocks in total. 

If two parity groups cross at one block, there is no need to send this block twice. As shown in Section \ref{def}, any horizontal and any diagonal parity group cross at a block, and each block can be the crossing of two groups at most once.  There are $x(p-1-x)$ crossings. The total number of blocks sent is 
\begin{align} \label{eq4}
\gamma & =  \underbrace{x p}_{\textrm{horizontal}}+ \underbrace{(p-1-x) (p-1)}_{\textrm{diagonal}}  + \underbrace{2}_{S_1} - \underbrace{x (p-1-x)}_{\textrm{crossings}} \nonumber \\
& =  (p-1)p + 2 - (x+1)(p-1-x)  \\
& \ge  (p-1)p +2- (p^2-1)/4  =  (3p^2-4p+9)/4 \nonumber 
\end{align}
The equality holds when $x = (p-1)/2$ or $x= (p-3)/2$, where $x$ is an integer.

This result states that we only need to send about $3/4$ of the total amount of information. And the slopes of the $n$ chosen parity groups do not matter as long as half are horizontal and half are diagonal. Moreover, similar repair bandwidth can be achieved using RDP or X-code. 
 For RDP code, the repair bandwidth is 
$$\frac{3(p-1)^2}{4}$$
which was also derived independently in \cite{optRDP}.
For X-code, the repair bandwidth is at most
$$\frac{3p^2-2p+5}{4}$$

The derivation for RDP is the following. For RDP code, the first $p-1$ columns are information. The $p$-th column is the horizontal parity. The $(p+1)$-th column is the slope 1 diagonal parity (including the $p$-th column). The diagonal starting at $a_{p,1}=0$ is not included in any diagonal parities. Suppose the first column is erased. Each horizontal or diagonal parity group will require $p-1$ blocks of transmission. Every horizontal parity group crosses with every diagonal parity group. Suppose $(p-1)/2$ horizontal parity groups and $(p-1)/2$ diagonal parity groups are transmitted. Then the total transmission is
$$\gamma = \underbrace{(p-1)(p-1)}_{p-1 \textrm{ parity groups} } - \underbrace{\frac{p-1}{2}\frac{p-1}{2}}_{\textrm{crossings}} = \frac{3(p-1)^2}{4}$$
This result is also derived independently in \cite{optRDP}.

The derivation for X-code is as follows. For X-code, the $(p-1)$-th row is the parity of slope -1, excluding the $p$-th row. And the $p$-th row is the parity of slope 1, excluding the $(p-1)$-th row. 
Suppose the first column is erased. First notice that for each parity group, $p-2$ blocks need to be transmitted.
To recover the parity block $a_{p-1,1}$, one has to transmit the slope -1 parity group starting at $a_{p-1,1}$. To recover the parity block $a_{p,1}$, the slope 1 parity group starting at $a_{p,1}$ must be transmitted. But it should be noted that by the construction of X-code, this slope 1 parity group essentially is the diagonal starting at $a_{p-1,1}$, except for the first element $a_{p,1}$. Zero crossings happen between two parity groups of slopes -1 and 1, starting at $a_{i,1}$ and $a_{j,1}$, if 
$$<i+j> = p-2 \textrm{ or } <i+j> =p$$
Each slope 1 parity group has no more than 2 zero crossings with the slope -1 parity groups.

Suppose we choose arbitrarily $(p-1)/2$ slope 1 parity groups and $(p-3)/2$ slope -1 parity groups for the information blocks in the first column. Then not considering the parity group containing $a_{p,1}$, the number of slope 1 and slope -1 parity groups are both $(p-1)/2$. Excluding zero crossings, each slope 1 parity group crosses with at least 
$$(p-1)/2-2 = (p-5)/2$$
slope -1 parity groups. The total transmission is
$$\gamma \le \underbrace{p(p-2)}_{p \textrm{ parity groups}} - \underbrace{\frac{p-1}{2}\frac{p-5}{2}}_{\textrm{crossings}} =  \frac{3p^2-2p+5}{4}$$

Also, equation (\ref{eq4}) is optimal in some conditions:

\begin{theorem}
The transmission bandwidth in (\ref{eq4}) is optimal to recover a systematic node for EVENODD if no linear combinations are sent except $\sum_{i=1}^{p-1} b_{i,v}$, for $v=0,1$.
\end{theorem}
\begin{proof} To recover a systematic node, say, the first node, parity blocks $b_{i,v}$, $i=1,2,\dots,p-1$ must be sent, where $v$ can be 0 or 1 for each $i$. This is because $a_{i,1}$ is only included in $b_{i,0}$ or $b_{i,1}$. Besides, given $b_{i,v}$, the whole parity group $B_{i,v}$ must be sent to recover the lost block. Therefore, our strategy of choosing $x$ horizontal parity groups and $p-1-x$ diagonal parity groups has the most efficient transmission. Finally, since (\ref{eq4}) is minimized over all possible $x$, it is optimal.
\end{proof}

The lower bound by (\ref{eq3}) is
$$\frac{Md}{(d-k+1)k}=\frac{M(n-1)}{(n-k)k} = \frac{p(p-1)(p+1)}{2p} = \frac{p^2-1}{2}$$
where $d=n-1$, $n=p+2$, $k=p$, and $M=p(p-1)$. It should be noted that (\ref{eq3}) assumes that each node sends the same number of blocks, but our method does not.

\begin{example}
Consider the EVENODD code with $p=5$ in Figure \ref{fig: 5evenodd}. For $1 \le i \le 4$, the code has information blocks $a_{i,j}$, $1 \le j \le 5$, and parity blocks $b_{i,v}$, $v=0,1$. Suppose the first column is lost. Then by (\ref{eq4}), we can choose parity groups $B_{1,0},B_{2,0},B_{3,1},B_{4,1}$. The blocks sent are: $\sum_{i=1}^{p-1} b_{i,0},\sum_{i=1}^{p-1} b_{i,1},b_{1,0},b_{2,0},b_{3,1},b_{4,1}$ from the parity nodes and $a_{1,2},a_{1,3},a_{1,4},a_{1,5},a_{2,2},a_{2,3},a_{2,4},$ $a_{2,5},a_{4,5},a_{3,2}$ from the systematic nodes. Altogether, we send 16 blocks, the number specified by (\ref{eq4}). We can see that $a_{1,3}$ is the crossing of $B_{1,0}$ and $B_{3,1}$. Similarly, $a_{1,4},a_{2,2},a_{2,3}$ are crossings and are only sent once for two parity groups. 
\end{example}

\begin{figure}
	\centering
		\includegraphics[width=0.3\textwidth]{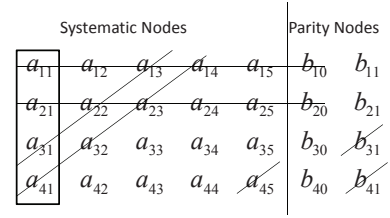}
	\caption{Repair of an EVENODD code with $p=5$. The first column is erased, shown in the box. 14 blocks are transmitted, shown by the blocks on the horizontal or diagonal lines. Each line (with wrap around) is a parity group. 2 blocks in summation form, $\sum_{i=1}^{p-1} b_{i,0},\sum_{i=1}^{p-1} b_{i,1}$ are also needed but are not shown in the graph.}
	\label{fig: 5evenodd}
	\vspace{-0.5cm}
\end{figure}

\section{$r$ Parity Nodes and One Erased Node}\label{section ex}
Next we discuss the repair of array codes with $r$ columns of parity, $r \ge 3$. And we consider the recovery in the case of one missing systematic column. In this section, we are going to use the extended EVENODD code \cite{evenoddex}, i.e. codes with parity columns of slopes $0,1,\dots,r-1$. Similar results can be derived for STAR code. Suppose the first column is erased without loss of generality. 

Let us first assume $r=3$, so the parity columns have slopes $0,1,2$. The repair strategy is: sending parity groups $B_{3n+v,v}$ for $v=0,1,2$ and $1 \le 3n+v \le p-1$. 

Let $A = \lfloor (p-1)/3 \rfloor$. Notice that $0 \le n \le A$ and each slope has no more than $\lceil(p-1)/3 \rceil$  but no less than $\lfloor (p-1)/3 \rfloor=A$ parity groups. 

Since there are three different slopes, there are crossings between slope 0 and 1, slope 1 and 2, and slope 2 and 0.
For any two parity groups $B_{i,1}$ and $B_{k,2}$, $<k-i> \neq 1$, so  (\ref{i+v}) does not hold. Hence no zero crossing exists for the chosen parity groups. Hence, every crossing corresponds to one block of saving in transmission. However, the total number of crossings is not equal to the sum of crossings between every two parity groups with different slopes. Three parity groups with slopes 0, 1, and 2 may share a common block, which should be subtracted from the sum.

Notice that the parity group $B_{i,v}$ contains the block $a_{i-vy,y+1}$. The modulo function ``$<>$'' is omitted in the subscripts. For three transmitted parity groups $B_{3n,0},B_{3m+1,1},B_{3l+2,2}$, if there is a common block in column $y+1$, then it is in row 
$3n \equiv 3m+1-y \equiv 3l+2-2y \quad (\textrm{mod } p)$.
To solve this, we get
$y \equiv 3(m-n)+1 \equiv 3(l-m)+1 \quad (\textrm{mod } p)$,
or $m-n \equiv l-m \quad (\textrm{mod } p)$. Notice $0 \le n, m, l <p/3$, so $-p/3 < m-n, l-m <  p/3 $. Therefore, $m-n = l-m$ without modulo $p$. Thus $l-n$ must be an even number. For fixed $n$, either $n \le m \le l \le A$, and there are no more than $(A-n)/2 + 1$ solutions for $(m,l)$; or $0\le l<m<n$, and the number of $(m,l)$ is no more than $n/2$. Hence, the number of $(n,m,l)$ is no more than 
$ \sum_{n=1}^{A} ((A-n)/2 + 1 + n/2) = A^2/2 + A$.

The total number of blocks in the $p-1$ chosen parity groups is less than $p(p-1)$. There are no less than $A$ parity groups of slope $v$, for all $0 \le v \le 2$, therefore for $0 \le u<v \le 2$, parity groups with slopes $u$ and $v$ have no less than $A^2$ crossings. Hence the total number of blocks sent in order to recover one column is:
\begin{eqnarray} \label{eq15}
 \gamma  & < & \underbrace{p(p-1)}_{p-1 \textrm{ parity groups}}-\underbrace{\binom{3}{2}
  A^2}_{\textrm{crossings}}+\underbrace{\frac{A^2+2A}{2}}_{\textrm{common}} +\underbrace{3}_{\sum_{i=1}^{p-1} b_{i,v}}  \nonumber \\
& < & \frac{13}{18}p^2 + \frac{17}{9} p - \frac{47}{18}  
\end{eqnarray}
where $(p-4)/3 < A \le (p-1)/3$. The above estimation is an upper bound because there may be better ways to assign the slopes of each parity group. Thus, we need to send no more than $13M/18$ blocks if $r=3$. 

By abuse of notation, we write $B_{m,v}=\{a_{<m+v(1-j)>,j}: j=2,\dots,p\}$ as the set of blocks (including the imaginary $p$-th row) in the parity group except $S_v$ and $a_{m,1}$. Let $M_v \subseteq \{1,2,\dots,q-1\}$, $0 \le v \le r-1$, be disjoint sets such that $\cup_{v=0}^{r-1} M_v = \{1,2,\dots,q-1\}$. Let $B_{M_{v},v} = \cup_{m \in M_{v}} B_{m,v}$. For given  $M_{v}$, define a function $f$ as 
$f(v_1,v_2,\dots,v_k) = |\{m_{1} \in M_{v_1}, \dots, m_{k} \in M_{v_{k}}: 
(m_2-m_1)/(v_2-v_1) \equiv (m_3-m_2)/(v_3-v_2) \equiv \dots (m_k-m_{k-1})/(v_k-v_{k-1}) \textrm{ mod } p \}|$,
for $k \ge 3$, and $0 \le v_1 < v_2 < \dots < v_{k} \le r-1$. Then we have the following theorem: 

\begin{theorem}
For the extended EVENODD with $r \ge 3$, the repair bandwidth for one erased systematic node is
\begin{eqnarray} \label{eq7}
\gamma & < & p(p-1)+p+r - \sum_{0\le v_1 < v_2 \le r-1} |M_{v_1}| |M_{v_2}| \nonumber \\
& & + \sum_{0 \le v_1 < v_2 < v_3 \le r-1}f(v_1,v_2,v_3)-\dots \nonumber \\
& &  + (-1)^{r-1}f(0,1,\dots,r-1)
\end{eqnarray}
\end{theorem}

\begin{proof}
Suppose the first column is missing and we transmit the parity groups $B_{m,v}$, $m \in M_{v}$ for $v=0,1,\dots,r-1$. 
Since the union of $M_{v}$ covers $\{1,2,\dots,q-1\}$, all the blocks in the first column can be recovered. The repair bandwidth is the cardinality of the union of $B_{M_v,v}$ plus the number of zero crossings and the summation blocks $\sum_{i=1}^{p-1}b_{i,v}$. The number of zero crossings is no more than the size of the imaginary row, $p$. The number of the summation blocks is $r$. 

By inclusion--exclusion principle, the cardinality of the union of $B_{M_v,v}$ is
\begin{eqnarray*}
 \sum_{0 \le v \le r-1}|B_{M_v,v}| - \sum_{0 \le v_1 < v_2 \le  r-1}|B_{M_{v_1},v_1} \cap B_{M_{v_2},v_2}| \\
 + \sum_{0 \le v_1 < v_2 < v_3 \le r-1}|B_{M_{v_1},v_1} \cap B_{M_{v_2},v_2}\cap B_{M_{v_3},v_3}|\\
 -\dots  + (-1)^{r-1} |B_{M_{0},0} \cap B_{M_{1},1} \dots B_{M_{r-1},r-1}|
\end{eqnarray*}

Every $|B_{m,v}| \le p$, so $\sum_{0 \le v \le r-1}|B_{M_v,v}| \le p(p-1)$.
Every two parity groups $B_{m_1,v_1}, B_{m_2,v_2}$ cross at a block. Hence $|B_{M_{v_1},v_1} \cap B_{M_{v_2},v_2}| = |M_{v_1}| |M_{v_2}|$. Since $B_{m,v}$ contains $a_{<m+v(1-j)>,j}$, $j=2,\dots,p$, the intersection of more than two parity groups $B_{m_1,v_1},\dots,B_{m_k,v_k}$ is equivalent to the solutions of 
$$m_1 - v_1 y \equiv m_2 - v_2 y \equiv \dots \equiv m_k - v_k y \textrm{ mod } p$$ 
where $y+1$ is the column index of the intersection. Or,
$$y \equiv \frac{m_2-m_1}{v_2-v_1} \equiv \dots \equiv \frac{m_k-m_{k-1}}{v_{k}-v_{k-1}} \textrm{ mod } p$$
Therefore, 
$$|B_{M_{v_1},v_1} \cap B_{M_{v_2},v_2}\cap \dots B_{M_{v_k},v_k}|= f(v_1,v_2,\dots,v_k)$$
And (\ref{eq7}) follows.
\end{proof}

We can see that (\ref{eq15}) is a special case of (\ref{eq7}), with $M_{v}=\{3n+v: 1 \le 3n+v \le p-1\}$, for $v=0,1,2$. For $r=4,5$, we can derive similar bounds by defining $M_v$.

Choose 
\begin{equation} \label{eq16}
M_{v}=\{rn+v: 1 \le rn+v \le p-1\}
\end{equation}
for $v=0,1,\dots,r-1$. Let $A=\lfloor(p-1)/r \rfloor$. 
And for $0 \le v_1<v_2<v_3 \le r-1$, $f(v_1,v_2,v_3)$ becomes the number of $(n_1,n_2,n_3)$, $1 \le r n_i + v_i \le p-1$, such that 
$$(n_2-n_1)(v_3-v_2) \equiv (n_3-n_2)(v_2-v_1) \textrm{ mod } p$$
Since $-p/r < n_2-n_1, n_3-n_2 < p/r $, and $(v_3-v_2)+(v_2-v_1) < r$, the above equation becomes
$$(n_2-n_1)(v_3-v_2) = (n_3-n_2)(v_2-v_1)$$
without modulo $p$. Therefore,
\begin{eqnarray*}
&&n_3-n_1 = (n_3-n_2)+(n_2-n_1) \\
& =& c \cdot \textrm{lcm}(v_3-v_2,v_2-v_1) 
\left( \frac{1}{v_3-v_2} + \frac{1}{v_2-v_1} \right) \\
& =& c \frac{v_3-v_1}{\textrm{gcd}(v_3-v_2,v_2-v_1)}
\end{eqnarray*}
where $c$ is an integer constant, lcm is the least common multiplier and gcd is the greatest common divisor. And for fixed $n_1$, the number of solutions for $(n_2,n_3)$ is no more than $1+(A-n_1)\textrm{gcd}(v_3-v_2,v_2-v_1)/(v_3-v_1)$, when $n_1 \le n_2 \le n_3 \le A$; and no more than $n_1\textrm{gcd}(v_3-v_2,v_2-v_1)/(v_3-v_1)$, when $0 \le n_3 < n_2 < n_1$. The number of $(n_1,n_2,n_3)$ is 
\begin{eqnarray*}
f(v_1,v_2,v_3) < \sum_{n_1} 1+(A-n_1+n_1)  \frac{\textrm{gcd}(v_3-v_2,v_2-v_1)}{v_3-v_1} \\
=A \left(1+A\frac{\textrm{gcd}(v_3-v_2,v_2-v_1)}{v_3-v_1} \right)
\end{eqnarray*}

Similarly, for four parity groups,
$$f(v_1,v_2,v_3,v_4) >
A \left(1+(A+2)\frac{\textrm{gcd}(v_4-v_3,v_3-v_2,v_2-v_1)}{v_4-v_1} \right)$$
For five parity groups, 
$$f(v_1,v_2,v_3,v_4,v_5) <
A+A^2\frac{\textrm{gcd}(v_5-v_4,v_4-v_3,v_3-v_2,v_2-v_1)}{v_5-v_1}$$ 

When $r=4$, equation (\ref{eq7}) becomes
\begin{eqnarray*}
\gamma & < & p(p-1)+p+4 - \sum_{0\le v_1 < v_2 \le 3} |M_{v_1}| |M_{v_2}| \nonumber \\
& & + \sum_{0 \le v_1 < v_2 < v_3 \le 3}f(v_1,v_2,v_3) -f(0,1,2,3)
\end{eqnarray*}
By the previous equations, 
$$f(0,1,2), f(1,2,3) < A(1+A/2)$$
$$f(0,1,3), f(0,2,3) < A(1+A/3)$$
$$f(0,1,2,3)>A(1+(A+2)/3$$ 
And the repair bandwidth is 
$$\gamma \approx p^2 - \binom{4}{2}(\frac{p}{4})^2+(2 \times \frac{1}{2}+2 \times \frac{1}{3})(\frac{p}{4})^2 - \frac{1}{3} (\frac{p}{4})^2 = \frac{7}{24}p^2$$ 
where the terms of lower orders are omitted.

When $r=5$, we can use (\ref{eq7}) again and get 
$$\gamma \approx p^2 +(-\binom{5}{2} + \frac{4}{2}+\frac{4}{3} + \frac{2}{4} - \frac{2}{3} - \frac{3}{4} + \frac{1}{4})(\frac{p}{5})^2 = \frac{53}{75}p^2$$
where the terms of lower orders are omitted.

It should be noted that the number of common blocks affects the bandwidth a lot. If we consider only the first 4 terms in (\ref{eq7}), any assignment of $M_v$ with equal sizes will result in a lower bound of $\gamma > (r+1)p^2/(2r) \approx p^2/2$, when $r$ is large. But due to the common blocks, the true $\gamma$ values for $r=4,5$ using (\ref{eq16}) has only slight improvement compared to the case of $r=3$.

The lower bound (\ref{eq3}) is 
$\frac{Md}{k(d-k+1)}=\frac{p(p-1)(p+r-1)}{pr} \approx \frac{p(p+r-1)}{r}$. When $r=3$, this bound is about $p^2/3$.

\section{3 Parity Nodes and 2 Erased Nodes} \label{section 2erase}

Up to now, we have considered the recovery problem given that one column is erased. Next, let us assume that two information columns are erased and we need to recover them successively. So we first recover one of the erased nodes, and then the other one. The first recovery is discussed in this section, and the second recovery was already discussed in the previous sections. Suppose we have 3 columns of parity with slopes -1, 0, and 1, which is in fact the STAR code in \cite{star}. Again, the arguments can be applied to extended EVENODD in a similar way. Without loss of generality, assume the first and $(x+1)$-th columns are missing, $1 \le x \le p-1$. 

Let  $B_{i,0}$,$B_{i,1}$, and $B_{i,-1}$ be $i$-th parity group of slopes 0, 1, and -1, respectively, $i = 1,2,\dots,p-1 $. The following are $3(p-1)/2$ parity groups that repair the first column:
$B_{0,-1}, B_{x,0},B_{2x,1},B_{2x,-1},B_{3x,0},B_{4x,1},\dots, \\B_{(p-3)x,-1},  B_{(p-2)x,0}, B_{(p-1)x,1}$.
For each parity block above, the corresponding recovered blocks are:
$a_{x,1+x},\\ a_{x,1},a_{2x,1},a_{3x,1+x}, a_{3x,1},a_{4x,1},\dots,a_{(p-2)x,1+x}, \\
a_{(p-2)x,1}, a_{(p-1)x,1}$.
An example of $p=5,x=1$ is shown in Figure \ref{fig:star}.

Rearrange the columns in the following order: Columns $1,1+x,1+2x,\dots,1+(p-1)x$ (every index is computed modulo $p$). 
We can see that the chosen parity groups $B_{jx,0}$, $j=x,3x,\dots,(p-2)x$ contain the blocks in Rows $Z=\{x,3x,\dots,(p-2)x\}$. $B_{jx,1}$ contains blocks $a_{jx,1},a_{(j-1)x,1+x},\dots,a_{(j-p+1)x,1+(p-1)x}$, for $j=2,4,\dots,p-1$. And similarly $B_{jx,-1}$ contains blocks $a_{jx,1},a_{(j+1)x,1+x},\dots,a_{(j+p-1)x,1+(p-1)x}$, for $j=0,2,\dots,p-3$. 

Now notice that the blocks included in the above parity groups have the $(1+x)$-th column as the vertical symmetry axis. That is, the row indices of the blocks needed in Columns $1$ and $1+2x$ are the same; those of Columns $1+(p-1)x$ and $1+3x$ are the same; ...; those of Columns $1+(p+3)x/2$ and $1+(p+1)x/2$ are the same. For example, the second column in Figure \ref{fig:star} is the symmetry axis. Thus, we only need to consider Columns $1+2x,1+3x,\dots,1+(p+1)x/2$.  

For columns $1+ix$, where $i$ is even and $2 \le i \le (p+1)/2 $, parity groups $\{B_{2x,1},B_{4x,1},\dots,B_{(p-1)x,1}\}$ include the blocks in Rows $X=\{2x,4x,\dots,(p-1-i)x\}$. And parity groups $\{B_{0,-1},B_{2x,-1},\dots,B_{(p-3)x,-1}\}$ include the blocks in Rows $Y=\{ix,(i+2)x,\dots,(p-1)x\}$. Since $2 \le i \le (p+1)/2 $, we have $i \le (p-1-i)+2$, and $X \cup Y = \{2x,4x,\dots,(p-1)x\}$. Hence $X \cup Y \cup Z = \{1,2,\dots,p-1\}$. Thus every block in Column $1+ix$ needs to be sent, for even $i$. 

Similarly, for Columns $1+ix$, where $i$ is odd and $3 \le i \le (p+1)/2$, parity groups $\{B_{2x,1},B_{4x,1},\dots,B_{(p-1)x,1}\}$ include the blocks in Rows $X=\{(p-i+2)x,(p-i+4)x,\dots,(p-1)x\}$. Parity groups $\{B_{0,-1},B_{2x,-1},\dots,B_{(p-3)x,-1}\}$ include the blocks in Rows $Y=\{2x,4x,\dots,(i-3)x\}$. Since $2 \le i \le (p+1)/2 $, we have $i-3< p-i+2$, and $X \cup Y = \{2x,4x,\dots,(i-3)x,(p-i+2)x,(p-i+4)x,\dots,(p-1)x\}$. Therefore, the rows not included in $X$ or $Y$ or $Z$ are $W=\{(i-1)x,(i+1)x,\dots,(p-i)x\}$ and $|W|=(p+3)/2-i$. The total saving in block transmissions for all the columns is:
$$2  \sum_{i \textrm{ odd, } 3\le i\le (p+1)/2} (\frac{p+3}{2}-i) = 
\left\{\begin{array}{ll} 
\frac{(p-1)^2}{8}, &  \frac{p+1}{2} \textrm{ odd} \\
\frac{(p+1)(p-3)}{8},&  \frac{p+1}{2} \textrm{ even} \\
\end{array} 
\right. $$

The above argument can be summarized in the following theorem.
\begin{theorem}
When two systematic nodes are erased in a STAR code, there exist a strategy that transmit about $7/8$ of all the information blocks, and about $1/2$ of all the parity blocks so as to recover one node.
\end{theorem}

The repair bandwidth $\gamma$ in the above theorem is about $7p^2/8$. Comparing it to the lower bound (\ref{eq3}), 
$\frac{Md}{k(d-k+1)} = \frac{p(p-1)(p+1)}{2p} \approx \frac{p^2}{2}$, 
we see a gap of $\frac{3p^2}{8}$ in total transmission. 

\begin{figure}
	\centering
		\includegraphics[width=0.2\textwidth]{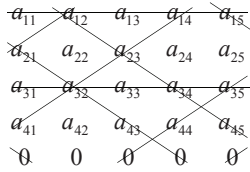}
		\caption{The recovery strategy for the first column in STAR code when the first and second columns are missing. $p=5$, $x=1$.}
	\label{fig:star}
	\vspace{-0.5cm}
\end{figure}

\section{Conclusions} \label{conclusion}
We presented an efficient way to repair one lost node in EVENODD codes and two lost nodes in STAR codes. Our achievable schemes outperform the naive method of rebuilding by reconstructing all the data. For EVENODD codes, a bandwidth of roughly $3M/4$ is sufficient to repair an erased systematic node. Moreover, if no linear combinations of bits are transmitted, the proposed repair method has optimal repair bandwidth with the sole exception of the sum of the parity nodes. Since array codes only operate on binary symbols, and our repair method involves no linear combination of content within a node except in the parity nodes, the proposed construction is computationally simple and also requires smaller disk I/O to read data during repairs. 

There are several open problems on using array codes for distributed storage. Although our scheme does not achieve  
the information theoretic cut-set bound, it is not clear if that bound is achievable for fixed code structures or limited field sizes. 
If we allow linear combinations of bits within each node, the optimal repair remains unknown. Our simulations indicate that shortening of EVENODD (using less than $p$ columns of information) further reduces the repair bandwidth  but proper shortening rules and repair methods need to be developed. Repairing other families of array codes or Reed-Solomon codes would also be of substantial practical interest.


\end{document}